\documentclass[12pt,a4paper]{article}
\usepackage{natbib}

\usepackage{authblk}
\usepackage[margin=1in]{geometry}

\usepackage{eucal}
\usepackage{amsmath, bbm, amsthm}
\usepackage{graphicx,psfrag,epsf}
\usepackage{enumerate}
\usepackage{url} 
\usepackage{xcolor}
\usepackage{caption}
\usepackage{algorithm}
\usepackage{algpseudocode}
\usepackage{enumitem}
\usepackage[scientific-notation=true]{siunitx}
\newtheorem{theorem}{Theorem}
\usepackage{amssymb}
\newcommand{\lamp}[0]{\eta_{\xi}}
\newcommand{\vola}[0]{\sigma_{\xi}}
\newcommand{\driftOrig}[0]{b_{\theta,\xi}}

\newcommand{\drift}[0]{\alpha_{\theta}}
\newcommand{\A}[0]{A_{\theta}}
\newcommand{\G}[0]{G_{\theta}}

\newcommand{\lowBd}[0]{l_{\theta}}
\newcommand{\upBd}[0]{r_{\theta}}

\newcommand{\obs}[0]{x^{\xi}}
\newcommand{\m}[0]{\mu^{\xi}}
\newcommand{\jointDist}[2]{\pi(\theta^{#1},\xi^{#1},\mathcal{S}(\mathcal{A})^{#2},\mathcal{D})}
\newcommand{\ud}{\,\mathrm{d}}

\graphicspath{{./plots/}}

\providecommand{\keywords}[1]
{
  \small	
  \textbf{\textit{Keywords---}} #1
}

\title{Flexible Bayesian inference for diffusion processes using splines}

\author[1,2,3]{Paul A. Jenkins}
\author[3,4]{Murray Pollock}
\author[1,3]{Gareth O. Roberts}
\affil[1]{{\small Department of Statistics, University of Warwick, Coventry, CV4 7AL, United Kingdom.}}
\affil[2]{{\small Department of Computer Science, University of Warwick, Coventry, CV4 7AL, United Kingdom.}}
\affil[3]{{\small The Alan Turing Institute, British Library, London, NW1 2DB, United Kingdom.}}
\affil[4]{{\small School of Mathematics, Statistics and Physics, Newcastle University, Newcastle-upon-Tyne, NE1 7RU, United Kingdom.}}

\begin{document}

\maketitle
\begin{abstract}
We introduce a flexible method to simultaneously infer both the drift and volatility functions of a discretely observed scalar diffusion. We introduce spline bases to represent these functions and develop a Markov chain Monte Carlo algorithm to infer, a posteriori, the coefficients of these functions in the spline basis. A key innovation is that we use spline bases to model transformed versions of the drift and volatility functions rather than the functions themselves. The output of the algorithm is a posterior sample of plausible drift and volatility functions that are not constrained to any particular parametric family. The flexibility of this approach provides practitioners a powerful investigative tool, allowing them to posit a variety of parametric models to better capture the underlying dynamics of their processes of interest. We illustrate the versatility of our method by applying it to challenging datasets from finance, paleoclimatology, and astrophysics. In view of the parametric diffusion models widely employed in the literature for those examples, some of our results are surprising since they call into question some aspects of these models.
\end{abstract}


\section{Introduction}

Diffusion processes have found wide application across the engineering, natural and social sciences \citep{kloeden1992numerical, van_zanten_nonparametric_2013}. For instance, they have been successfully used in neuroscience to model the membrane potential of neurons \citep{lansky2008review}, in molecular dynamics to model the angles between atoms evolving in a force field \citep{papaspiliopoulos2012nonparametric}, and in astrophysics to describe quasar variability over time \citep{kelly2009variations}. Within econometrics they have modelled asset prices and interest rates, and have been used to price financial instruments \citep{karatzas1998methods}. Other applications include paleoclimatology, where they are used to model glacial cycles in energy balance models and variability in the intensity of El Ni\~{n}o and the Southern Oscillation \citep{imkeller_conceptual_2002}. 

In general all of these applications aim to infer from discrete temporal observations $\mathcal{D} = \{v_{t_i}\}_{i=0}^N$ the coefficients of an underlying diffusion model, $V$, that is a Markov solution to the stochastic differential equation (SDE)
\begin{equation}\label{eq:master_SDE}
  \ud V_t = b(V_t) \ud t + \sigma(V_t) \ud W_t,\quad t\in[0,T],\quad V_0=v_0,
\end{equation}
where $b$ and $\sigma$ are the drift and volatility coefficients. 

If we restrict ourselves to parametric models for $b$ and $\sigma$ (i.e.\ the functional forms of the coefficients are known up to the value of a parameter vector $\theta$) then inference is still challenging. In particular, for a given parametrization the likelihood of the observations $\mathcal{D}$ will typically not be known in closed form, and as a consequence a considerable literature has developed to tackle this problem. Within a frequentist paradigm, approaches include estimating functions \citep{bibby2010estimating}, maximization of approximate likelihood functions \citep{dacunha1986estimation, ait2002maximum, ait2008closed}, and simulation-based schemes \citep{pedersen1995new,durham_numerical:2002,beskos2006exact,beskos2009monte}. In the Bayesian literature \cite{roberts2001inference} proposed a Markov chain Monte Carlo (MCMC) approach using data-augmentation, also utilized by \cite{bladt2014simple, golightly_bayesian:2008, sermaidis_markov:2013}, and \cite{ van2017bayesian}.

However, for many practical problems the functional form of $b$ and $\sigma$ are either unknown or disputed within the applied literature. Specification of a good finite-dimensional model is challenging, particularly outside the context of physical phenomena in the natural sciences. In finance for instance, finding realistic models is particularly difficult, as exemplified by disputes surrounding models for interest rates or stock prices \citep{bali_comprehensive_2006, durham_numerical:2002}. In such a situation a non-parametric method of inference may be more attractive as we need make no restrictive assumptions about functional forms for the drift and diffusion coefficients. From the perspective of a practitioner these methods are particularly appealing as the functional form that arises can be used to identify a plausible and interpretable parametric family, or to gain direct insight about the dynamics of the underlying process.

The frequentist, non-parametric literature is dominated by kernel-type estimators. Examples include those of \cite{banon_nonparametric_1978,stanton_nonparametric_1997}, locally linear smoothers with adaptive bandwidth \citep{spokoiny_adaptive_2000}, and estimators derived via penalized likelihood \citep{comte_penalized_2007}. Showing consistency and contraction rates of the estimators is, however, non-trivial \citep{dalalyan2002asymptotically, gobet2004nonparametric, tuan_nonparametric_1981, zanten_rates_2001}. 

The literature on Bayesian, non-parametric inference for diffusions is not as well developed as its frequentist counterpart \citep[see][for an overview]{van_zanten_nonparametric_2013}, and to our knowledge methods of inference only for the drift coefficient have been described to date in the setting of low-frequency observations. In \cite{papaspiliopoulos2012nonparametric} the unknown drift function is equipped with a prior measure in function space, which is assumed to be Gaussian with mean $0$ and covariance defined by a certain differential operator, and a data augmentation scheme \citep{roberts2001inference} is then used to compute an MCMC approximation to the posterior. Consistency results for this setting are shown by \cite{pokern2013posterior}, and improved contraction rates are derived by \cite{waaij_gaussian_2016}. \cite{van_der_meulen_reversible_2014} proposed an algorithmic modification to the procedure of \cite{papaspiliopoulos2012nonparametric}, using a different basis expansion for the drift function and employing random truncation of this expansion (with a truncation point equipped with a prior and explored with a reversible jump step). Contraction rates for this approach were derived in \cite{van_der_meulen_adaptive_2018}. Some additional references establishing consistency and contraction rates include \cite{gugushvili_nonparametric_2014, koskela2018consistency, nickl_nonparametric_2017,nic:ray:2020, meulen_consistent_2013}. Finally, \cite{gugushvili2023nonparametric} recently proposed a Bayesian method of inference of the diffusion coefficient for high frequency financial datasets.

In the Bayesian setting when dealing with real-world datasets of discretely observed diffusions, the simultaneous estimation of \emph{both} $b$ and $\sigma$ in \eqref{eq:master_SDE} is a challenging problem, and one which current Bayesian non-parametric approaches can not address. In this paper we propose a flexible Bayesian algorithm for simultaneous estimation of both drift and diffusion coefficients for discretely observed diffusions, without any restriction on the observation frequency, drawing on the strengths of both parametric and non-parametric paradigms. Our approach is parametrized in a way that is flexible in adapting to subtle patterns in the data, yet once a set of hyperparameters is fixed the model becomes parametric and consistency results and simplicity of implementation of the regular parametric approach apply. The method has the additional advantage that it avoids the need to work with a discretised version of \eqref{eq:master_SDE}. Substantial emphasis of our work is put into achieving an efficient, self-contained, and user-friendly algorithm for inference on the functional form of $b$ and $\sigma$ for SDEs. Visualization of the functional form of $b$ and $\sigma$ is a powerful investigative tool for practitioners, allowing them to better understand the dynamics of their process of interest, and give them insight as to what may be appropriate parametric models.

Our method has two crucial components. The first component is the introduction of a spline basis \citep{de1978practical} to model (indirectly) $b$ and $\sigma$. Splines are compactly supported piece-wise polynomial bases which offer us a great deal of flexibility in modelling functions, and critically are able to capture their local behaviour. The second component is that a scalar diffusion $V$ can be equivalently defined either via the pair of functions $(b(\cdot),\sigma(\cdot))$ or via $(A(\cdot),\eta(\cdot))$, where
\begin{equation}\label{eq:lamperti}
\begin{split}
\eta(v)&:=\int^v_0\frac{1}{\sigma(u)} \ud u,\\
A(x)&:=\int^x_0\alpha(u) \ud u,\quad \alpha(x) = \frac{b(\eta^{-1}(x))}{\sigma(\eta^{-1}(x))}-\frac{1}{2} \sigma'(\eta^{-1}(x)).
\end{split}
\end{equation}
$\eta$ is commonly known as the Lamperti transformation \citep{lamperti1964simple}. In this paper we model the drift and volatility functions by expressing $A$ and $\eta$ in spline bases. The Lamperti transformation is crucial in avoiding potential issues of degeneracy in the methodology we subsequently develop. Briefly, the issue is as follows. Conceptually we are augmenting the parameter space, to be explored by an MCMC algorithm, with the path space which takes the entire sample path of the diffusion as a latent variable. However, the volatility coefficient is completely determined by this sample path via its quadratic variation, at least in the range of the path. Thus a Gibbs-type algorithm that attempts to alternate between updates of the sample path and updates of $\sigma$ is `reducible': the sample path allows for only one possible $\sigma$ and $\sigma$ can no longer be updated; neither can any new sample paths compatible with other choices for $\sigma$ be proposed. One solution is to propose updates not to the latent sample path but to its underlying driving Brownian motion, which can be done without reference to any parameters determining $\sigma$; the Lamperti transformation arises naturally in constructing the mapping between Brownian motion and the original diffusion. See Appendix \ref{apx:supplement} and extensive discussion in \citet{roberts2001inference} for further details. Once one accepts the need to parametrise $V$ using $\eta$, then $A$ (or equivalently $\alpha$) is the only free function remaining.

An additional benefit of the parameterization in \eqref{eq:lamperti} is that it is possible to enforce monotonicity of $\eta$ in the spline basis (via so-called \emph{I-splines}), and thus guarantee $\sigma$ be positive. In view of the definition of $\eta$ as an integrated positive function this property is essential. Without this reparametrization, a direct spline representation of $\sigma$ could become negative in some regions of the state space. As discussed in our methodological sections the transformation \eqref{eq:lamperti}, together with the availability of easily computable bounds on functions built with splines, have the additional advantage that we can avoid any time-discretization of \eqref{eq:master_SDE}.

Within the framework of a spline basis representation of $A$ and $\eta$, we provide an MCMC algorithm for sampling from the posterior of the basis parameters, given $\mathcal{D}$. We proceed via data-augmentation \citep{roberts2001inference, sermaidis_markov:2013}, alternately updating the basis parameters and the latent sample path connecting observations in $\mathcal{D}$. Under conditions on $b$ and $\sigma$ which we show to hold in their implied spline representations, it is in fact possible to implement an algorithm using only a finite-dimensional surrogate for each sample path, circumventing the need to discretize the model.

We benchmark the performance of our algorithm and its accuracy in recovering the coefficients of a true, generating SDE on an illustrative example. We then apply the method to three real-world datasets: from finance, on the evolution of the short-term interest rate through  three-month treasury bills; from paleoclimatology, looking at the fluctuations of historical temperatures on the Northern Hemisphere; and from astrophysics, examining quasar light variability. These showcase a broad range of potential applications. The results obtained on the financial dataset are largely in agreement with the conclusions drawn from the use of competing, frequentist methods. However, for the two other examples we show that the parametric models commonly used in the literature may need revisiting.

This paper is organized as follows: in Section \ref{sec:main_section} we present our flexible family of diffusions, together with our choice of spline bases; in Section \ref{sec:inference_algo} we then develop an MCMC algorithm (Algorithm \ref{alg:mcmc}) targeting the coefficients of the spline bases; in Section \ref{sec:compconsiderations} we discuss a number of practical considerations in the implementation of Algorithm \ref{alg:mcmc}, including the particular choices of knots together with their location, and regularization with appropriate choices of prior; in Section \ref{sec:numerical} we consider our methodology applied to the broad range of real-world examples we discussed above; finally, in Section \ref{sec:conclusions} we outline natural continuations of our work from both a methodological and application perspective. Technical details where appropriate are collated in the appendix. 

\section{A flexible family of diffusions}\label{sec:main_section}

As noted in the introduction, we first apply the transformation \eqref{eq:lamperti} to represent $b:\mathbbm{R}\to\mathbbm{R}$ and $\sigma:\mathbbm{R}\to\mathbbm{R}_{\geq 0}$ as $A:\mathbbm{R}\to\mathbbm{R}$ and $\eta:\mathbbm{R}\to\mathbbm{R}_{\geq 0}$. We then represent $A$ and $\eta$ as follows:
\begin{equation}\label{eq:lamperti_splines}
\A(\cdot)=\sum_{i=0}^{M_\theta}\theta_i u_i(\cdot) = \theta^Tu(\cdot),\quad \lamp(\cdot)=\sum_{i=0}^{M_\xi}\xi_i (h_i(\cdot)-h_i(\bar{v})) = \xi^Th(\cdot)-\xi^Th(\bar{v}),
\end{equation}
where $u = (u_1,\dots,u_{M_{\theta}})$ and $h = (h_1,\dots,h_{M_{\xi}})$ are two, possibly distinct, sets of (twice-differentiable) basis functions, $\theta$ and $\xi$ are vectors of coefficients, and $\bar{v}\in\mathbbm{R}$ is a free parameter used for centering $\lamp$ (more details about $\bar{v}$ are given in Section \ref{sec:compbases}). We take a Bayesian approach: given a choice of bases $u$ and $h$ and a prior $\pi(\theta,\xi)$ on the parameters of interest, our interest is in the posterior $\pi(\theta,\xi | \mathcal{D})$. We will develop an MCMC algorithm targeting this distribution.

An additional aim is to avoid any form of time-discretization of the SDE in \eqref{eq:master_SDE}, circumventing the need to analyse any introduced bias. Using a data-augmentation scheme constitutes one way to achieve this; however, as we shall see, there are a series of practical problems that first need to be overcome. In particular it must be possible to compute a series of quantities in closed form: the Jacobian of the Lamperti transformation, $D_\xi=1/\sigma_\xi$; the integrand in the exponent of a Radon--Nikodym derivative of the law of diffusion $X$ with respect to Wiener law, $G_\theta:=\left(\alpha_\theta^2+\alpha_\theta'\right)/2$; and global upper and lower bounds on $G_\theta$. The first two are easy to derive with any choice of (sufficiently differentiable) basis functions and are given by
\begin{equation}\label{eq:generic_eq_semi_param}
    D_\theta(\cdot)=\xi^Th'(\cdot),\quad G_\theta(\cdot) = \frac{1}{2}\left(\left(\theta^T u'(\cdot)\right)^2 + \theta^Tu''(\cdot)\right),
\end{equation}
where $'$ applied to a vector denotes component-wise derivative. It is the need to compute global upper and lower bounds on $G_\theta$ that substantially narrows down the possible choices for the basis $u$. On one hand, so long as all of $u_i$ $(i=1,\dots, M_{\theta})$ have bounded first and second order derivatives, i.e.\ $\lVert u'\rVert_\infty:=\sup\{ \lvert u'_i(y)\rvert;y\in\mathbbm{R}, i\in\{1,\dots,M_{\theta}\}\}<\infty$ and $\lVert u''\rVert_\infty<\infty$, it is always possible to bound $G_\theta$ by:
\begin{equation}\label{eq:generic_bounds_on_G}
    -\frac{1}{2}\sum_{i=0}^{M_\theta}\lvert\theta_i\rvert\lVert u''_i\rVert_\infty\leq G_\theta \leq \frac{1}{2}\left(\sum_{i=0}^{M_{\theta}}\lvert\theta_i\rvert \lVert u'_i\rVert_\infty\right)^2+\frac{1}{2}\sum_{i=0}^{M_\theta}\lvert\theta_i\rvert\lVert u''_i\rVert_\infty.
\end{equation}
On the other hand, the bounds above are almost always too crude and render the algorithm impossible to implement in practice. The choice of basis $u$ is therefore dictated by the need for tight bounds on $G_\theta$.

In this paper we propose to use splines as bases $u$ and $h$. Splines address all of the issues discussed above. In particular, unlike general bases for which bounds on $\G$ need to be computed from \eqref{eq:generic_bounds_on_G} (and are thus limited by the computational issues that arise from the explosion of those bounds with the number of included basis terms), the maxima and minima of $\G$ can be efficiently identified in the spline basis. Ultimately this yields tight bounds regardless of the number of included basis functions, and offers more degrees of freedom than many alternatives. Additionally, it is possible to regularize splines, exerting direct control over the desired level of flexibility of $\A$ and $\lamp$ (see Section \ref{sec:compknots} for details). Finally, it is possible to directly equip any function expanded in a spline basis with a monotonicity property, substantially reducing the size of the function space for $\lamp$ a priori; this is crucial in light of \eqref{eq:lamperti} with $\eta$ an integral of a positive function.

Naturally, a piecewise polynomial basis can be defined in many ways. Perhaps the most common family of splines are the so-called \emph{B-splines} \citep{de1978practical}. B-splines are piecewise polynomial curves with compact support for which numerically efficient algorithms exist. For us they are a natural choice as we can simply choose the compact support to include the range of the data plus some margin to avoid edge effects. This allows us to model a flexible class of functions, and gives us analytical tractability (including derivatives) and a user-selectable degree of smoothness. 

B-splines are controlled by two sets of hyper-parameters: basis order (signifying the maximal order of any polynomial used) and the placement and number of \emph{knots} (positions at which different polynomial basis' elements are spliced together). The two jointly control the maximal flexibility of functions that can be modelled using the chosen basis. In principle any desired degree of flexibility can be achieved by simply fixing the basis order and the number and density of knots to high enough levels \citep{micula2012handbook}. For robustness and stability of coefficients the De Boor's recursion formula \citep{de1978practical} defining B-splines is most commonly used in practice. We follow this convention to define the basis $u$, and define the basis $h$ via related I-splines to guarantee monotonicity. We describe both spline types in detail below.

\textbf{B-splines} (or related M-Splines) are defined by the number and locations of \emph{knots}, as well as the order of the polynomials used. Let $\kappa_i \in \mathbbm{R}$ $(i=1,\dots,\mathcal{K})$ denote the knots (heuristically these are the positions at which splines are anchored), and let $B_i(x|k)$ denote the $i$th B-spline of the $k$th order evaluated at $x$. The notation for M-spline is defined analogously. M-splines and B-splines are defined by the recurrence relations:
\begin{equation}\label{eq:M_and_B_splines}
  \begin{split}
    M_i(x|0) &:= \frac{1}{\kappa_{i+1}-\kappa_i}\mathbbm{1}_{[\kappa_i,\kappa_{i+1})}(x),\\
    M_i(x|k) &:= \frac{(k+1)[(x-\kappa_i)M_i(x|k-1) + (\kappa_{i+k+1}-x)M_{i+1}(x|k-1)]}{k(\kappa_{i+k+1}-\kappa_i)},\, k=1,2,\dots;\\
    B_i(x|k) &= (\kappa_{i+k+1}-\kappa_i)M_i(x|k)/(k+1),\quad k=0,1,\dots;\; i=1,\dots,\mathcal{K}.
  \end{split}
\end{equation}

\textbf{I-splines} are integrated M-splines \citep{he1998monotone,ramsay1988monotone} and thus the $i$th I-spline of the $k$th order evaluated at $x$ is given by:
\begin{align}
    I_i(x|k) &:= \int_{\kappa_0}^x M_i(u|k)du\notag\\
    &\phantom{:}=  \sum_{j=0}^{M_\kappa-1}\mathbbm{1}_{[\kappa_j,\kappa_{j+1})}(x)\bigg(\mathbbm{1}_{(-\infty,j-k)}(i) 
    + \mathbbm{1}_{[j-k,j]}(i)\sum_{m=i}^j\frac{\kappa_{m+k}-\kappa_m}{k+2}M_m(x|k+1)\bigg).\label{eq:I_splines}
\end{align}
Setting $h_i(\cdot):=I_i(\cdot|k)$ and restricting the coefficients $\xi$ to be non-negative guarantees $\xi^Th(\cdot)$ to be monotonically increasing. This restriction can be imposed easily by exponentiating coefficients $\xi$, i.e.\ rather than defining $\lamp$ via \eqref{eq:lamperti_splines} we write, with abuse of notation:
\begin{equation}
  \lamp(\cdot)=\sum_{i=0}^{\mathcal{K}}e^{\xi_i} (h_i(\cdot)-h_i(\bar{v})) =: \left(e^{\xi}\right)^T\left(h(\cdot)-h(\bar{v})\right). \label{eq:exponentiation}
\end{equation}
where exponentiation of a vector is taken to be componentwise.

Between each pair of knots $[\kappa_i,\kappa_{i+1})$, a spline is simply a polynomial. Consequently, to find tight bounds on $\G$ we can break up the domain of $\A$ into intervals $[\kappa_i,\kappa_{i+1})$ $(i=1,\dots,\mathcal{K}-1)$, so that $\G$ becomes a polynomial on each of them, and then use standard methods to find local extrema of $\G$. Global bounds $\lowBd\leq\G$ and $\upBd\geq\G-\lowBd$ are then given by the minimum (resp.\ maximum) of all local bounds. Numerical root-finding schemes might in principle give inaccurate results; however, there exist methods that quantify the error bounds and an additional margin can then be added to offset possibly incurred errors \citep{rump2003ten}. 

\section{Inference algorithm}\label{sec:inference_algo}

Following the introduction of our family of diffusions in Section \ref{sec:main_section}, in this section we develop an MCMC algorithm targeting the posterior for $\theta$ and $\xi$, the coefficients of the spline bases. We use a data-augmentation scheme in which the unobserved parts of the paths $\mathcal{A}:=\{V_t;t\in(t_i,t_{i+1}), i =0,\dots,N-1\}$ are treated as missing values. By targeting the joint posterior distribution $\pi(\mathcal{A},\theta,\xi|\mathcal{D})$ via Gibbs sampling, the distribution of interest, $\pi(\theta,\xi|\mathcal{D})$, would be admitted as a marginal. However, as noted by \citet{roberts2001inference}, a na\"ive augmentation scheme $(\theta,\xi,\mathcal{D})\rightarrow(\theta,\xi,\mathcal{D},\mathcal{A})$ would cause two problems: (i) it prompts for the imputation of an infinite-dimensional object $\mathcal{A}$, which is obviously impossible to achieve on a computer; and (ii) it leads to a chain which cannot mix because $\mathcal{A}$ fully determines the diffusion coefficient. Any update of $\xi$ conditioned on $\mathcal{A}$ is then degenerate as the conditional density is given by a point mass at a current value of $\xi$.

We have set up our spline basis carefully so that these problems can be solved, as follows. First we employ the transformation $\eta$ as defined in \eqref{eq:lamperti}. It can be shown that $X_t = \eta(V_t)$ is a diffusion satisfying
\begin{equation}
    \label{eq:Lamperti_SDE}
    \ud X_t = \alpha(X_t) \ud t + \ud W_t,\quad t\in[0,T],\quad X_0=x_0 := \eta(v_0),
\end{equation}
in particular $X$ has \emph{unit} volatility independent of the parameters $(\theta,\xi)$ \citep{roberts2001inference}. Second, we augment the parameters not with $\mathcal{A}$ but instead with the alternative: $(\theta,\xi,\mathcal{D})\rightarrow(\theta,\xi,\mathcal{D},\mathcal{S}(\mathcal{A}))$, where $\mathcal{S}(\mathcal{A})$ denotes a finite-dimensional \emph{surrogate} for the unobserved path $\mathcal{A}$. By an appropriate choice of surrogate, we can ensure that it is possible to sample $\mathcal{S}(\mathcal{A})$ without having to discretize time as would be necessary when sampling $\mathcal{A}$ directly, and we can also ensure that the output is an almost surely finite-dimensional random variable (termed a \emph{skeleton}) from which an entire path $\mathcal{A}$ can be reconstructed if needed. We use a \emph{path-space rejection sampler} with proposals based on Brownian bridges \citep{beskos2008factorisation}. This specific approach (in which a skeleton is used) was suggested and developed in \citet{beskos2006exact} and \citet{sermaidis_markov:2013}. 

Inference is then performed by Gibbs-type updates, alternately (a) updating unknown parameters by drawing from $\pi(\theta,\xi|\mathcal{S}(\mathcal{A}),\mathcal{D})$ and then (b) imputing the unobserved path by sampling from $\pi(\mathcal{S}(\mathcal{A})|\theta,\xi,\mathcal{D})$. As required, the marginal distribution of the parameter chain converges to the posterior distribution $\pi(\theta,\xi|\mathcal{D})$. For step (a) we use a Metropolis--Hastings step, following \cite{sermaidis_markov:2013} (see Appendix). Step (b) is where we employ path-space rejection sampling. We summarise the approach in Algorithm \ref{alg:mcmc}, which we will now describe in detail.

The algorithm accepts as input $\kappa^{(\theta)}$ and $\kappa^{(\xi)}$, used to denote the vectors of knots for drift and volatility coefficients respectively, and $\mathcal{O}^{(\theta)}$ and $\mathcal{O}^{(\xi)}$, which denote the orders of the respective bases. Computations performed for steps \ref{algStep:psrs} and \ref{algStep:mhprob} depend directly on this quadruplet of fixed parameters.

\begin{algorithm}[t]
\caption{Flexible and Exact Bayesian MCMC for discretely observed diffusions}\label{alg:mcmc}
\begin{algorithmic}[1]
\Require{$\kappa^{(\theta)},\kappa^{(\xi)},\mathcal{O}^{(\theta)}, \mathcal{O}^{(\xi)}, M, (\theta^{(0)}, \xi^{(0)})$}
\Ensure{$\{(\theta^{(n)},\xi^{(n)})\}_{n=1}^M$}

\For{$n=1,\dots, M$}
\State Draw $\mathcal{S}(\mathcal{A})^{(n)}\sim\pi(\mathcal{S}(\mathcal{A})|\theta^{(n-1)},\xi^{(n-1)},\mathcal{D})$\label{algStep:psrs}\Comment{See Appendix \ref{apx:supplement}}
\State Draw $(\theta^{\circ},\xi^{\circ})\sim q((\theta^{(n-1)},\xi^{(n-1)}),\cdot)$ \label{algStep:paramupdate}\Comment{As per \eqref{eq:qstep1} and \eqref{eq:qstep2}}
\vspace{0.5em}
\State Set $a\leftarrow1\wedge\dfrac{q((\theta^{\circ},\xi^{\circ}),(\theta^{(n-1)},\xi^{(n-1)}))\pi(\theta^{\circ},\xi^{\circ},\mathcal{S}(\mathcal{A})^{(n)},\mathcal{D})}{q((\theta^{(n-1)},\xi^{(n-1)}),(\theta^{\circ},\xi^{\circ}))\pi(\theta^{(n-1)},\xi^{(n-1)},\mathcal{S}(\mathcal{A})^{(n)},\mathcal{D})}$\label{algStep:mhprob}
\vspace{0.5em}
\If{$U <a$, with $U\sim \texttt{Unif}[0,1]$}
\State Set $(\theta^{(n)},\xi^{(n)})\leftarrow (\theta^{\circ},\xi^{\circ})$
\Else
\State Set $(\theta^{(n)},\xi^{(n)})\leftarrow (\theta^{(n-1)},\xi^{(n-1)})$
\EndIf
\EndFor
\end{algorithmic}
\end{algorithm}

First considering the parameter update step (Algorithm \ref{alg:mcmc} Step \ref{algStep:paramupdate}), often it is not possible to sample from $\pi(\theta,\xi|\mathcal{S}(\mathcal{A}),\mathcal{D})$ directly. (More precisely, this conditional density is given up to a constant by the right-hand side of \eqref{eq:joint_param_sampler}, treated as a function of $(\theta,\xi)$ only. Evidently the dependence on $\theta$ and $\xi$ is rather complicated in general.) However, because the joint density $\pi(\theta,\xi,\mathcal{S}(\mathcal{A}),\mathcal{D})$ can be computed in closed form (see \eqref{eq:joint_param_sampler} and \citet[Thm 3]{sermaidis_markov:2013}), it is possible to employ a Metropolis--Hastings correction. To update a large number of parameters at once we further exploit gradient information to improve the quality of proposals. In particular, we employ a Metropolis-adjusted Langevin algorithm \citep[MALA;][]{roberts1996exponential} and update all coordinates of $\xi$ and $\theta$ at once by defining the proposal $q((\theta^{(n-1)},\xi^{(n-1)},\cdot)$ via
\begin{align}
\xi^{\circ}&\sim\mathcal{N}\left(\xi^{(n-1)}+\delta_1 \nabla_{\xi}\log\left(\jointDist{(n-1)}{(n)}\right); \delta_2^2 I\right), \label{eq:qstep1}\\
\theta^{\circ}&\sim\mathcal{N}\left(\theta^{(n-1)}+\delta_3 \nabla_{\theta}\log\left(\pi(\theta^{(n-1)},\xi^{\circ},\mathcal{S}(\mathcal{A})^{(n)},\mathcal{D})\right); \delta_4^2 I\right), \label{eq:qstep2}
\end{align}
where $I$ is the identity matrix of appropriate size and $\delta_1,\dots,\delta_4$ are tuning parameters. To compute $\nabla_\theta\log(\jointDist{}{})$ and $\nabla_\xi\log(\jointDist{}{})$ in a spline context we require only the quantities $\nabla_\theta\lowBd$ and $\nabla_\theta\upBd$. It is not always possible to find the closed form expressions for those two quantities, but instead a finite difference scheme can be employed. 

Now consider the imputation of the unobserved path (Algorithm \ref{alg:mcmc} Step \ref{algStep:psrs}). The following arguments may be found in \citet{beskos2005exact}, \citet{beskos2006retrospective}, \citet{beskos2008factorisation}, and \citet{sermaidis_markov:2013}, and we give only a brief summary. To explain the idea, suppose for the moment that we have discrete observations directly from \eqref{eq:Lamperti_SDE}; that is, $\lamp$ is simply the identity and we are interested in inference of $\theta$ only. To impute the finite-dimensional surrogate variable $\mathcal{S}(\mathcal{A})$ it is now enough to employ $N$ independent path-space rejection samplers, each for a separate interval $(t_i,t_{i+1})$, $i=0,\dots,N-1$, so for simplicity consider a single interval $[0,t]$ with $X_0 = x$ and $X_t = y$. Denoting the law of the bridge $(X|X_0 = x, X_t = y)$ under \eqref{eq:Lamperti_SDE} by $\mathbbm{P}^{(t,x,y)}$ and a Brownian bridge connecting the same points by $\mathbbm{W}^{(t,x,y)}$, we have that
\begin{equation}
    \label{eq:RND}
\frac{d\mathbbm{P}^{(t,x,y)}}{d\mathbbm{W}^{(t,x,y)}} \propto \exp\left(-\int_0^t [\G(X_s) - \lowBd]\ud s\right).
\end{equation}
We recognise the right-hand side as the probability that a Poisson process of unit intensity on $[0,t]\times[0,\upBd]$ has zero points beneath the graph of $s\mapsto \G(X_s) - \lowBd$, allowing for a rejection sampler to be implemented using Brownian bridge proposals and with acceptance probability \eqref{eq:RND} (even though the expression is intractable with finite resources), an example of \emph{retrospective simulation}. This motivates the choice of $\mathcal{S}(\mathcal{A})$ as
\[
\mathcal{S}(\mathcal{A}) := \{\{Z_{\psi_{j}},\{\psi_{j}, \chi_{j}\}\}_{j=1}^{\varkappa},\varkappa\},
\]
where $\{\psi_{j}, \chi_{j}\}_{j=1}^{\varkappa}$ is a unit-intensity Poisson Point Process on $[0,t]\times[0,\upBd]$ and $Z \sim \mathbbm{W}^{(t,x,y)}$. (By convention $\{\cdot\}_{j=1}^0:=\emptyset$.) By inspection of \eqref{eq:RND}, this choice of $\mathcal{S}(\mathcal{A})$ can now be simulated by rejection: (i) Simulate the Poisson process $\{\psi_{j}, \chi_{j}\}_{j=1}^{\varkappa}$, (ii) Simulate $Z$ at the times $\psi_1,\dots,\psi_{\varkappa}$; (iii) Accept $\mathcal{S}(\mathcal{A})$ if $\G(Z_{\psi_j})-\lowBd < \chi_j$ for each $j=1,\dots,\varkappa$.

In the more general case, when we have discrete observations from a diffusion of the form \eqref{eq:master_SDE} rather than \eqref{eq:Lamperti_SDE}, considerable further complication is introduced by the fact that, following an application of the Lamperti transformation, the datapoints $\{\lamp(v_{t_i}):\: i=1,\dots,N\}$ now depend on a parameter of interest. Brownian bridges of the form  $\mathbbm{W}^{(t_{i+1}-t_i,\lamp(v_{t_i}),\lamp(v_{t_{i+1}}))}$ are now inapplicable as dominating measure in each interval. This issue is resolved by a further re-parametrization which in some contexts is known as \emph{noise outsourcing}; details are given in the Appendix.

We now provide a formal verification that path-space rejection sampling theory can be applied in our spline context, which suffices for path-space rejection sampling within MCMC.

\begin{theorem}\label{thm:amenability_of_splines}
Let a diffusion model be defined by the Lamperti transformation $\lamp$ and an anti-derivative of a drift of a Lamperti-transformed diffusion $\A$ as in \eqref{eq:lamperti}. Suppose further that $\lamp$ and $\A$ can be expanded in I-spline and B-spline bases as in \eqref{eq:lamperti_splines}, with bases orders fixed to $\mathcal{O}^{(\xi)}\geq 3$ and $\mathcal{O}^{(\theta)}\geq 3$. Then it is possible to simulate exactly from $\pi(\mathcal{S}(\mathcal{A})|\theta,\xi,\mathcal{D})$.
\end{theorem}
\begin{proof}
Denote the domain over which $\lamp$ is defined with $\mathcal{R}^{(\xi)}:=[\kappa^{(\xi)}_1, \kappa^{(\xi)}_{\mathcal{K}^{(\xi)}}]$ and similarly the domain over which $\A$ is defined with $\mathcal{R}^{(\theta)}:=[\kappa^{(\theta)}_1, \kappa^{(\theta)}_{\mathcal{K}^{(\theta)}}]$. The algorithm never evaluates $\lamp$, $\A$ nor any of their derivatives outside of these two intervals (see Section \ref{sec:compbases}) and consequently the extensions of $\lamp$ and $\A$ to $\mathbbm{R}$ can be assumed to satisfy all the relevant regularity conditions outside of $\mathcal{R}^{(\xi)}$ and $\mathcal{R}^{(\theta)}$.

It follows directly from the definitions \eqref{eq:M_and_B_splines} and \eqref{eq:I_splines} that the $i^{th}$ order I-spline defines a $\mathcal{C}^{i}$ function (on $\mathcal{R}^{(\xi)}$) and $i^{th}$ order B-spline defines a $\mathcal{C}^{i-1}$ function (on $\mathcal{R}^{(\theta)}$). Consequently, with the choices $\mathcal{O}^{(\xi)}\geq 3$ and $\mathcal{O}^{(\theta)}\geq 3$, $\lamp$ is at least $\mathcal{C}^{3}$ on $\mathcal{R}^{(\xi)}$ and $\A$ is at least $\mathcal{C}^{2}$ on $\mathcal{R}^{(\theta)}$. By the construction in \eqref{eq:exponentiation} we know $\lamp'(v)$ is positive, and since
\begin{equation}\label{eq:transf_original_coefs}
\vola(v) = (\lamp'(v))^{-1},\quad \driftOrig(v) =\frac{\A'(\lamp(v))}{\lamp'(v)} - \frac{1}{2}\frac{\lamp''(v)}{(\lamp'(v))^3},
\end{equation}
it follows that $\vola$ is at least $\mathcal{C}^{2}$ on $\mathcal{R}^{(\xi)}$ and $\driftOrig$ is at least $\mathcal{C}^{1}$ on $\mathcal{R}^{(\theta)}$. This implies that $\driftOrig$ and $\vola$ are both locally Lipschitz and since their extensions to $\mathbbm{R}$ are arbitrary it follows that the SDE defined indirectly through $\A$ and $\lamp$ admits a unique solution \citep[Sec.\ 8.2]{karatzas1998brownian}.

To sample from $\pi(\mathcal{S}(\mathcal{A})|\theta,\xi,\mathcal{D})$ imposes additional conditions \citep{beskos2006exact}: $\drift$ must be at least $\mathcal{C}^{1}$, $\exists\,\widetilde{A}_\theta$ such that $\widetilde{A}_\theta'=\lamp$, $\exists\,\lowBd>-\infty$ such that $\lowBd\leq\inf_{u\in\mathbbm{R}}\G(u)$, and  $\exists\,\upBd<\infty$ such that $\upBd\geq\sup_{u\in\mathbbm{R}}\G(u)-\lowBd$. Clearly, as $\A$ is at least $\mathcal{C}^{2}$ on $\mathcal{R}^{(\theta)}$, $\drift:=\A'$ is at least $\mathcal{C}^{1}$. Additionally, by construction $\widetilde{A}_\theta=\A$ and the existence of the requisite bounds follows for instance from \eqref{eq:generic_bounds_on_G}.
\end{proof}

We remark that our choice of spline bases yields a relatively simple form for $\mathcal{S}(\mathcal{A})$; in applications of rejection sampling of diffusions elsewhere it is often necessary to simulate additional information about the diffusion such as its local extrema. Here, the availability of tight global bounds on $G_\theta$ as a consequence of our choice of spline bases in Section \ref{sec:main_section} obviate this complication.

\section{Practical considerations} \label{sec:compconsiderations}
\subsection{Choice of bases}
\label{sec:compbases}
The choice of bases can be regarded as a choice of functional prior on $\lamp(\cdot)$ and $\A(\cdot)$. Bases defined on a compact interval $C$ correspond to priors which are supported only on functions vanishing outside $C$. In principle, $C$ can be made arbitrarily large, eliminating the influence of the truncation of the prior's support. This however requires identification of the regions over which $\A$ and $\lamp$ need to be evaluated. We refer to those two regions as $C({\A})$ and $C({\lamp})$ respectively. In the case of $\lamp$ the procedure is simple: we need only ever evaluate $\lamp$ on the range of observations, so we can make an empirical choice $C({\lamp}):=[\min\{\mathcal{D}\}-\delta, \max\{\mathcal{D}\}+\delta]$, where $\delta\geq 0$ is some margin that allows us to avoid edge effects from the usage of splines.

Determining $C(A_\theta)$ is more difficult because region over which $A_\theta$ is evaluated depends on $\eta_\xi$. We take a pragmatic approach and try to identify a $C(A_\theta)$ which is large enough to draw the same inferential conclusions. Noting that there is little value in making $C(A_\theta)$ so large that it includes basis elements $u$ defined over regions without observations, this gives us a natural way to proceed. We begin by centering $C({\A})$ around the origin by using $\bar{v}$ in \eqref{eq:lamperti_splines} as an anchor. For simplicity we set $\bar{v}:=\frac{1}{N+1}\sum_{i=0}^Nv_{t_i}$, noting that the choice of anchor is one of convenience and any choice will lead to the same results. Next we initialise $C({\A})\leftarrow[-\mathcal{R},\mathcal{R}]$, for some $\mathcal{R}\in\mathbbm{R}_+$, and simply proceed with Algorithm \ref{alg:mcmc}. If at any point $\A$ (or any of its derivatives) needs to be evaluated outside of $[-\mathcal{R},\mathcal{R}]$ we halt Algorithm \ref{alg:mcmc}, double $\mathcal{R}$, and re-execute. 

\subsection{Placement of knots}
\label{sec:compknots}
We now consider the choice of locations and total number of knots as well as the order of polynomials. These award different degrees of flexibility to $\lamp$ and $\A$. A standard approach when using splines is to keep the order of polynomials moderate \citep[anything beyond third order is rarely used;][Sec.\ 5.2]{friedman2001elements}. Once $C(\A)$ and $C(\lamp)$ have been settled, it is \emph{a priori} reasonable to space knots equally across these intervals, except at the boundaries where knots can be duplicated to relax any continuity requirements there \citep[see Appendix to Ch.\ 5 of][]{friedman2001elements}. We take this approach to knot placement throughout our experiments; indeed, it can be viewed as an advantage of our method that good results can be obtained without the need to first fine-tune knot placement. Similarly, a user would ideally not want to have to perform extensive experimentation to determine the \emph{number} of knots. A Bayesian approach to this issue is to allow the user to specify too many knots and to employ a prior which induces an appropriate regularization. We use a Gaussian process prior on the integrated squared derivatives of the fitted function (in this case $\A$ and $\lamp$):
\begin{equation}\label{eq:penalty_on_wiggliness}
 \log\pi(\A, \lamp)\propto -\frac{1}{2}\sum_{k=1}^{K}\left(\lambda_{1,k} \int_{\mathbbm{R}}\left(\A^{(k)}(x)\right)^2dx+\lambda_{2,k} \int_{\mathbbm{R}}\left(\lamp^{(k)}(x)\right)^2dx\right).
\end{equation}
Here, $\cdot^{(k)}$ denotes $k$th order derivative (with respect to $x$) and $\lambda_{i,k}$, $i=1,2$; $k=1,\dots,K$, are the tuning hyper-parameters. For splines, the integrals above become:
\begin{equation*}
  \int_{\mathbbm{R}}\left(\A^{(k)}(x)\right)^2dx=\theta^T\Omega_k\theta,\quad \int_{\mathbbm{R}}\left(\lamp^{(k)}(x)\right)^2dx=\left(e^{\xi}\right)^T\widetilde{\Omega}_k\left(e^{\xi}\right),
\end{equation*}
where matrices:
\begin{equation}
\left\{\Omega_k\right\}_{ij}=\int_{\mathbbm{R}}\left(u^{(k)}_i(x)\right)\left(u^{(k)}_j(x)\right)dx,\quad \left\{\widetilde{\Omega}_k\right\}_{ij}=\int_{\mathbbm{R}}\left(h^{(k)}_i(x)\right)\left(h^{(k)}_j(x)\right)dx,
\end{equation}
are available in closed forms \cite[Sec.\ 5.4]{friedman2001elements}.

We found that additional, stronger prior information is required for the basis functions supported \emph{on the edges} of the intervals $C(\A)$ and $C(\lamp)$ (i.e.\ intervals $[\kappa_{i},\kappa_{i+1})$ with extreme values of $i$). For these, it is possible that only a few observations (or in extreme cases none) fall on the interior of $[\kappa_{i},\kappa_{i+1})$. Consequently, to guarantee the convergence of the Markov chains we impose an additional prior for these functions and shrink the respective $\theta_i$ and $\xi_i$ towards $0$. This is accomplished by defining diagonal matrices $P$ and $\widetilde{P}$ with non-negative diagonal elements, where large values indicate strong shrinkage of respective basis element towards $0$. An additional penalty $\theta^TP\theta+(e^\xi)^T\widetilde{P}(e^\xi)$ can then be added on to \eqref{eq:penalty_on_wiggliness}.

As a result, we end up with a prior of the form:
\begin{equation}
\label{eq:main_prior}
  \pi(\theta,\xi)\propto\exp\left\{ -\frac{1}{2}\theta^T\left(P+\sum_{k=1}^K\lambda_{1,k}\Omega_k\right)\theta \right\}\exp\left\{-\frac{1}{2}\left(e^{\xi}\right)^T\left(\widetilde{P}+\sum_{k=1}^K\lambda_{2,k}\widetilde{\Omega}_k\right)\left(e^{\xi}\right)\right\}.
\end{equation}

We have chosen for simplicity a prior for which $\theta$ and $\xi$ are independent. For example, placing independent priors on parameters of $b$ and $\sigma$ (rather than of $\eta$ and $A$) would, from \eqref{eq:lamperti}, induce priors on $\theta$ and $\xi$ which are not necessarily independent. However, if the data contains evidence of some unassumed dependence then this should of course ultimately reveal itself in the posterior. It would be straightforward to incorporate any prior knowledge about a correlation between $\theta$ and $\xi$ into the prior if desired; Algorithm \ref{alg:mcmc}, and specifically equation \eqref{eq:joint_param_sampler}, does not rely on a product form for $\pi(\theta,\xi)$. Similarly, one could consider priors other than the Gaussian form appearing in \eqref{eq:main_prior} though this could potentially introduce computational costs elsewhere. For example, a Gaussian prior ensures that the derivatives appearing in the MALA updates \eqref{eq:qstep1}--\eqref{eq:qstep2} remain well-behaved while other priors may not.

We follow a number of heuristics to reduce the dimensionality of the hyper-parameters in \eqref{eq:main_prior} following \citet{friedman2001elements}. In practice it is often sufficient to penalize only one integrated derivative of the $k^*$-th order and set other $\lambda_{i,k}=0$ $(k\neq k^*)$. Additionally, only extreme entries on the diagonals of $P$ and $\widetilde{P}$ need to be set to non-zero values and the algorithm is often quite robust to the actual values chosen. Consequently, the problem of parameter tuning is often reduced to dimension 3--4, and the final search for the most fitting values for the hyper-parameters can be completed by validation; that is, by splitting the dataset into training and testing parts, training the model on the former, evaluating the likelihood on the latter, and keeping the model with the highest likelihood attained on the test dataset.

\subsection{Computational cost}
\label{sec:compcost}
The computational cost of the algorithm will depend on all of its parameters and hyperparameters in a complicated way in general. However, we can pick out the main influences on this cost by noting that the most expensive part of Algorithm \ref{alg:mcmc} is typically Step \ref{algStep:psrs}, which employs a rejection sampler for each of the $N$ inter-observation intervals in order to simulate a set of skeleton points of a diffusion bridge. For a diffusion $X$ satisfying \eqref{eq:Lamperti_SDE} with say $X_{0} = x$ and $X_{\Delta_i} = y$, the acceptance probability for a proposed surrogate $\mathcal{S}(\mathcal{A})$, when using $\mathbbm{W}^{(\Delta_i,x,y)}$ as a proposal law, is
\begin{equation}
    \label{eq:acceptance}
\mathbbm{E}_{\mathbbm{W}^{(\Delta_i,x,y)}}\left[\exp\left(-\int_{0}^{\Delta_i}[\G(X_s) - \lowBd] ds\right)\right].
\end{equation}
This probability decays exponentially in $\Delta_i$; thus, we should expect the efficiency of the algorithm to diminish exponentially in the observation spacing. Conversely, as $\Delta_i$ decreases the latent bridge better resembles a Brownian bridge and the acceptance probability \eqref{eq:acceptance} goes to 1 as $\Delta_i \to 0$. It is further worth noting that, owing to the Markov property of the diffusion, the bridge between each pair of observations can be treated independently. Therefore for a fixed observation spacing the computational cost of Step \ref{algStep:psrs} is at most $O(N)$ in the number of observations $N$ as $N\to\infty$. If there are opportunities to exploit parallelization in the implementation then this cost can be reduced further. One can exploit the linear cost in $N$ to counteract the exponential cost in $\Delta_i$ by imputing additional datapoints between existing observations; see \citet[Section 3.4 and Section 4]{sermaidis_markov:2013} for this and other strategies on boosting efficiency, and \citet{pel:rob:2012} for an extensive empirical study.

\section{Numerical Experiments}\label{sec:numerical}

\subsection{Illustrative dataset} \label{sec:illustrative}
We begin our numerical experiments by considering an illustrative dataset to determine whether our methodology can recover the (known) underlying generative mechanism. We simulate 2001 equally spaced observations with inter-observation distance set to $0.1$ from the SDE:
\begin{equation}\label{eq:toy_data_SDE}
  \ud V_t = -V_t(V_t^2-1)\ud t + 1/(1+V_t^2) \ud W_t,\quad V_0=1,\quad t\in[0,200].
\end{equation}

\begin{figure}[ht]
    \centering
    \includegraphics[width=0.8\textwidth]{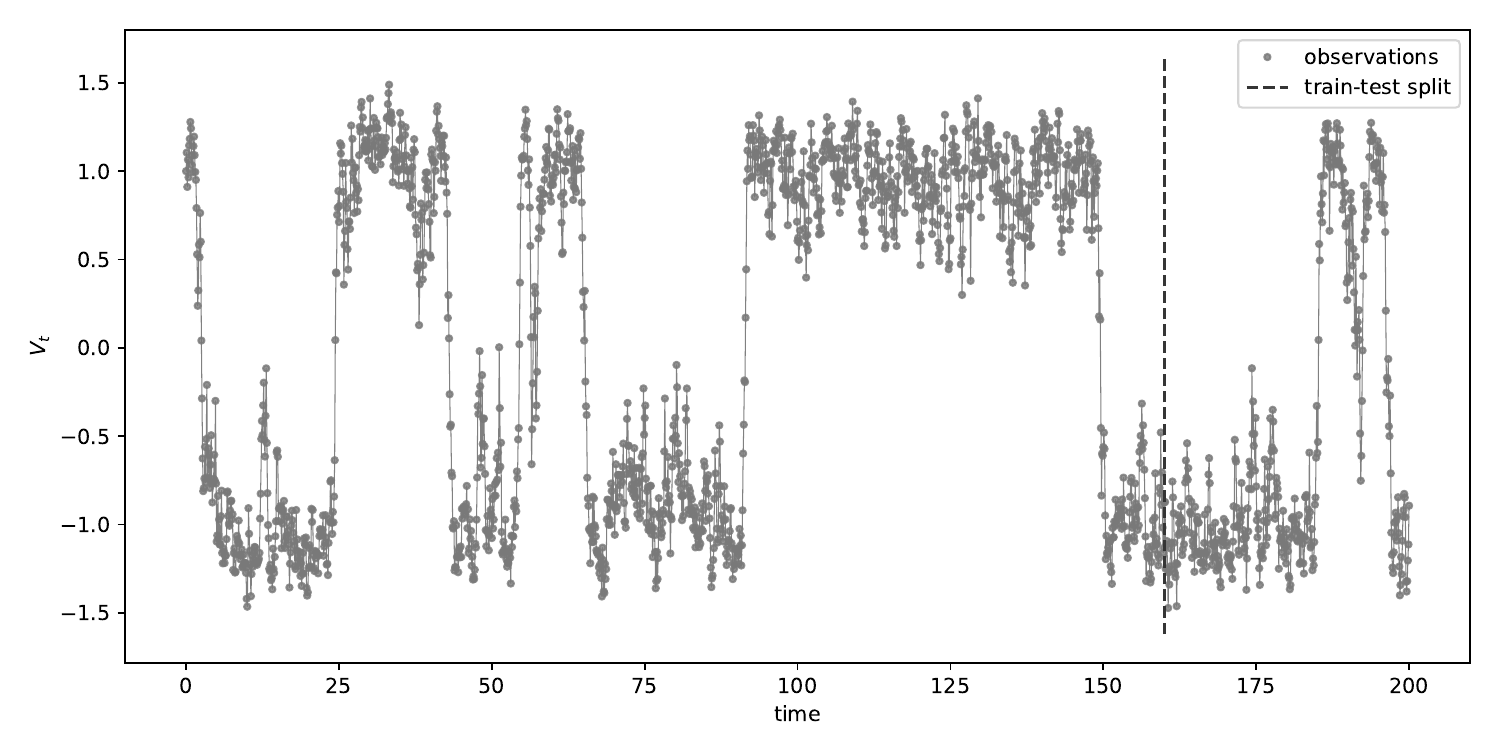}
    \caption{Observations of the underlying process, generated according to the SDE \eqref{eq:toy_data_SDE}. Dashed line indicates the split between the training (left) and test (right) dataset.\label{fig:toy_example_obs}}
\end{figure}

The simulated data is plotted in Figure \ref{fig:toy_example_obs}. (Recall the partition of the data into separate training and testing parts for the purpose of hyperparameter tuning; see Section \ref{sec:compknots}.) The observations range between $-1.53$ and $1.51$. Because of the well-behaved form of the drift and diffusion coefficient we could set the total number of knots to a moderate value and closely recover the two functions without resorting to strong priors. However, because in general we might have no prior information about the underlying process we follow the general principles presented in Section \ref{sec:compconsiderations} of over-specifying the total number and density of knots and relying on the regularization property of priors to see how faithfully the truth can be recovered. For bases $h$ we set $15$ equidistant knots between $-2$ and $2$, and further place three additional knots on each extreme value (in total there are four knots on $-2$ and four on $2$). The addition of extra margins on both sides of the observed range increases the flexibility of the function $\lamp$ near the end-points, contributing to a speed-up of Algorithm \ref{alg:mcmc}. Similarly, for basis $u$ we set $15$ equidistant knots between $-4$ and $4$ and place an additional $4$ knots on each of the extreme values (on $-4$ and $4$).

The order of the polynomial basis $h$ is set to $3$, which results in splines of the $3$rd order approximating the function $1/\sigma_\xi(\cdot)$. The order of basis $u$ is set to $4$ so that splines approximating the drift function $\drift(\cdot)$ are also of 3rd order, under the reasoning that the functions $1/\vola(\cdot)$ and $\drift(\cdot)$ could be expected to have similar smoothness properties \emph{a priori}. With these orders all derivatives needed by the algorithm for various computations exist and do not vanish. Additionally, the choice is in agreement with a common principle of keeping the order of the polynomials moderate \cite[Sec.\ 5.4]{friedman2001elements}. $\bar{v}$ was set to $0$. We ran the algorithm for $M=600,000$ iterations exploring various choices of hyper-parameters, specifically $\lambda_{1,3} \in \{10^{-3},10^{-2},\dots,10^2\}$ and $\lambda_{2,2} \in \{10^{-4},10^{-3},\dots,10^2\}$, and the final estimates (resulting in the highest values of averaged, noisy estimates of the likelihood) are presented in Figure \ref{fig:toy_example_fits}. Here $\lambda_{1,3}=0.1$, $\lambda_{2,2}=0.1$ were used for the prior and all other $\lambda_{i,k}$ were set to zero.

Recall that the algorithm aims to directly infer $\A$ and $\lamp$ ($\drift$ and $\lamp$ are given in the top row of Figure \ref{fig:toy_example_fits}). The posterior draws closely resemble the true functions $\alpha$ and $\lamp$ on the range of observations. Functions $\driftOrig$ and $\vola$ are computed as byproducts from $\A$ and $\lamp$ via identities \eqref{eq:transf_original_coefs} and hence small deviations of the latter from the truth result in larger deviations of $\driftOrig$ and $\vola$. This, together with the less dense observations around the origin, is the reason why the posterior draws of the drift coefficient exhibit increased uncertainty in this region. 

\begin{figure}[!t]
    \centering
    \includegraphics[width=\textwidth]{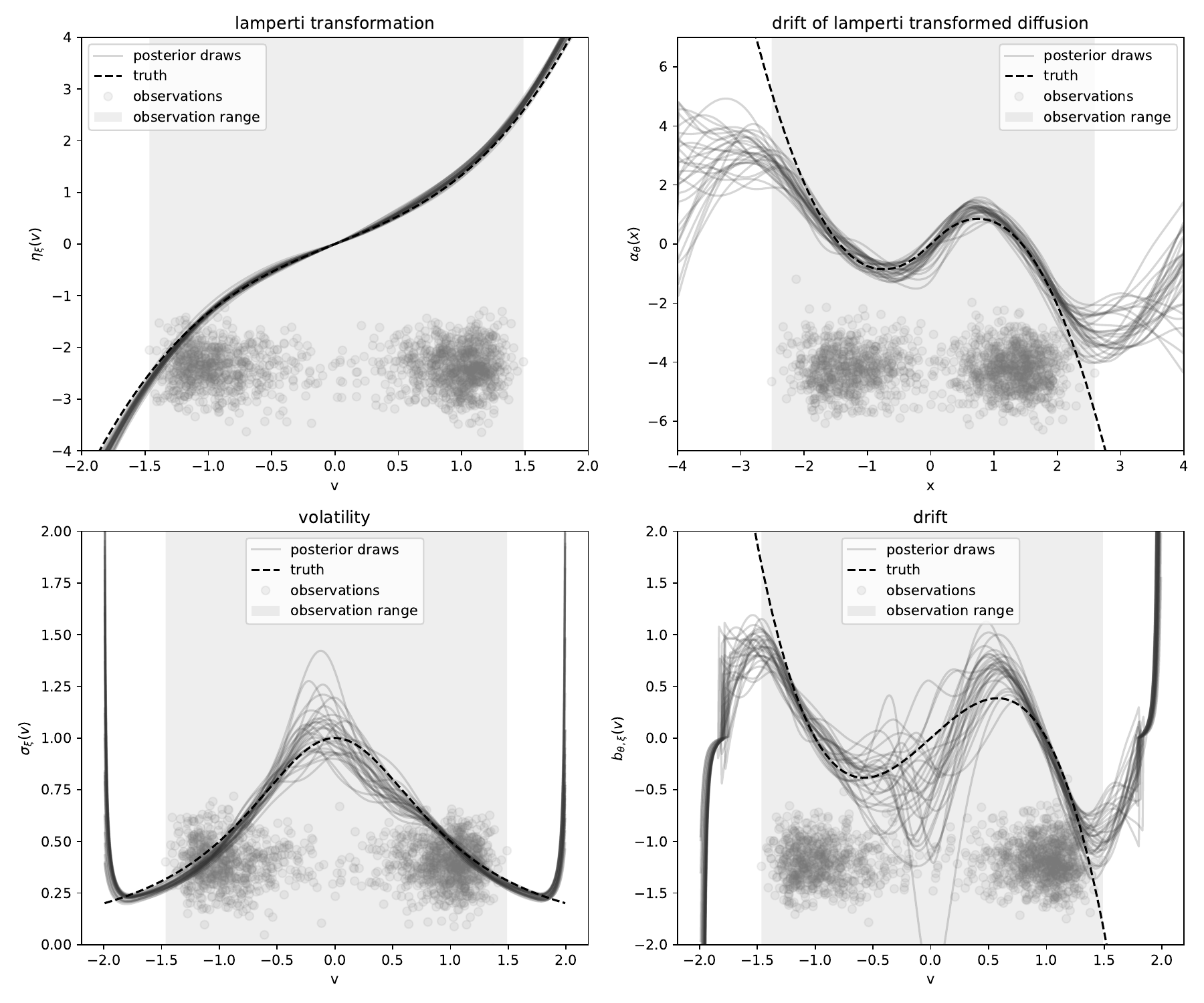}
    \captionof{figure}{Fits to the illustrative dataset. Dashed lines represent the true functions used to generate the data. The MCMC chain was run for 600,000 iterations and 30 random draws from the last 300,000 steps were plotted (solid curves). The upper right plot is on the Lamperti-transformed axis, $x=\lamp(v)$. Clouds of observations are plotted to illustrate the differences in the amount of data falling in different regions of space. For the observation scatterplot, the $y$-axis has no direct interpretation and is simply a jitter added to aid visualization.
    \label{fig:toy_example_fits}}    
\end{figure}

\subsection{Finance dataset}

In this section we consider modelling the U.S.\ short-term riskless interest rate. \cite{bali_comprehensive_2006} review some of the methods proposed in the econometrics literature for modelling this process. Following \cite{stanton_nonparametric_1997}, we use the three-month U.S. Treasury bills (T-bills) as its proxy. For the sake of fair comparison, we use the same test dataset as \cite{stanton_nonparametric_1997}: daily recordings of U.S. Treasury bills' rate from 1965 to 1995 (Figure \ref{fig:finance_example_obs}).

\begin{figure}[ht]
  \begin{center}
    \includegraphics[width=0.8\textwidth]{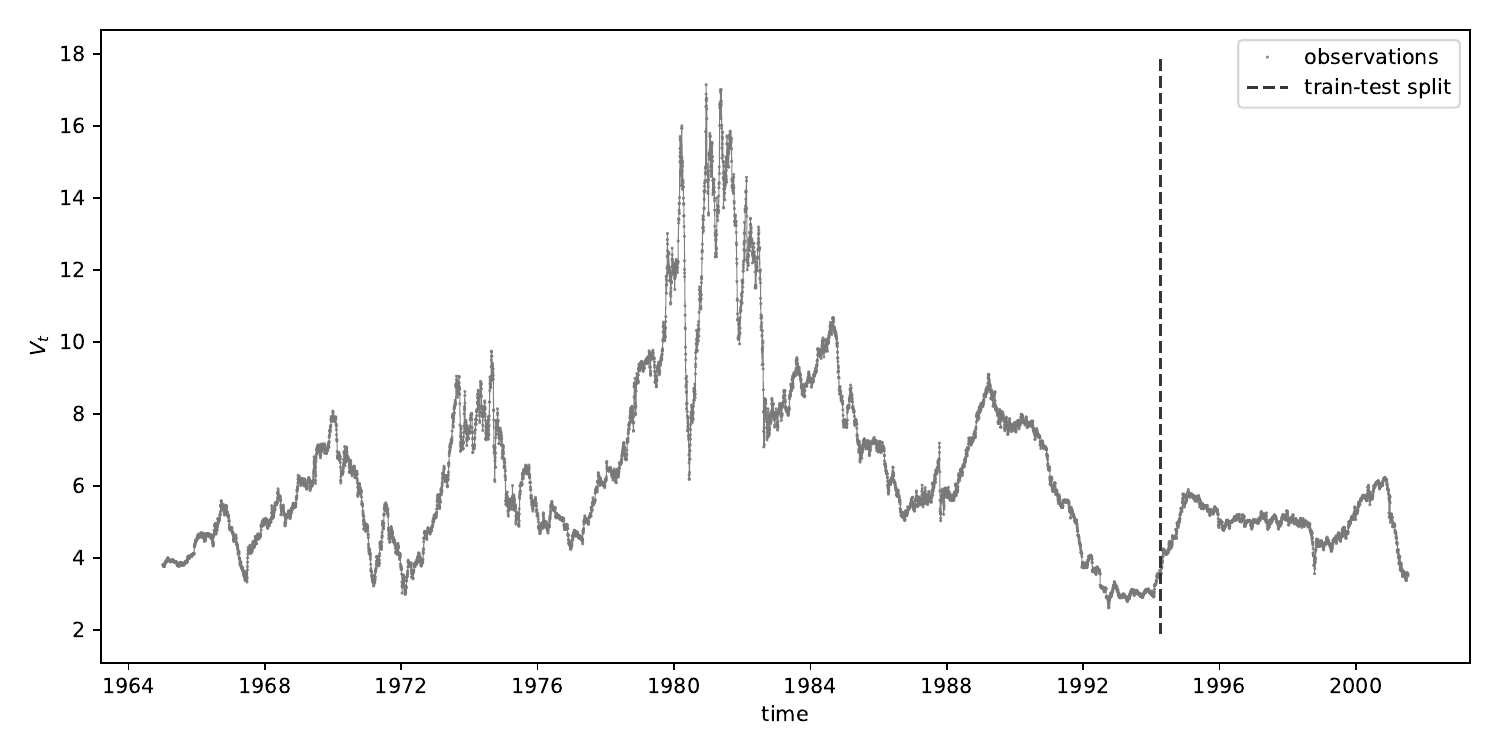}
    \caption{Daily observations of the three-month U.S. Treasury bills' rates 1965--2001.\label{fig:finance_example_obs}}
  \end{center}
\end{figure}

Following our methodology, we set the polynomial orders of bases $u$ and $h$ to $4$ and $3$ respectively, following the reasoning as in Section \ref{sec:illustrative}. We also set the knots of basis $u$ and $h$ respectively to:
\begin{align*}
  \kappa^{({\theta})}&\leftarrow\{-15_{(4)},-10.5, -6, -1.5, 3, 7.5, 12, 16.5, 21, 25.5, 30_{(4)}\},\\
    \kappa^{({\xi})}&\leftarrow\{0_{(4)},2,4,6,8,10,12,14,16,18,20_{(4)}\},
\end{align*}
where $_{(n)}$ in the subscript denotes the multiplicity of a knot (which in absence of the subscript is by default set to $1$). $\bar{v}$ was set to $5$. Tuning parameters were set to $\lambda_{1,3}=6000$, $\lambda_{2,2}=1$ (with all others set to $0$) and Algorithm \ref{alg:mcmc} was run for 200,000 iterations. The results are given in Figure \ref{fig:finance_fits}.

\begin{figure}[!t]
    \includegraphics[width=\textwidth]{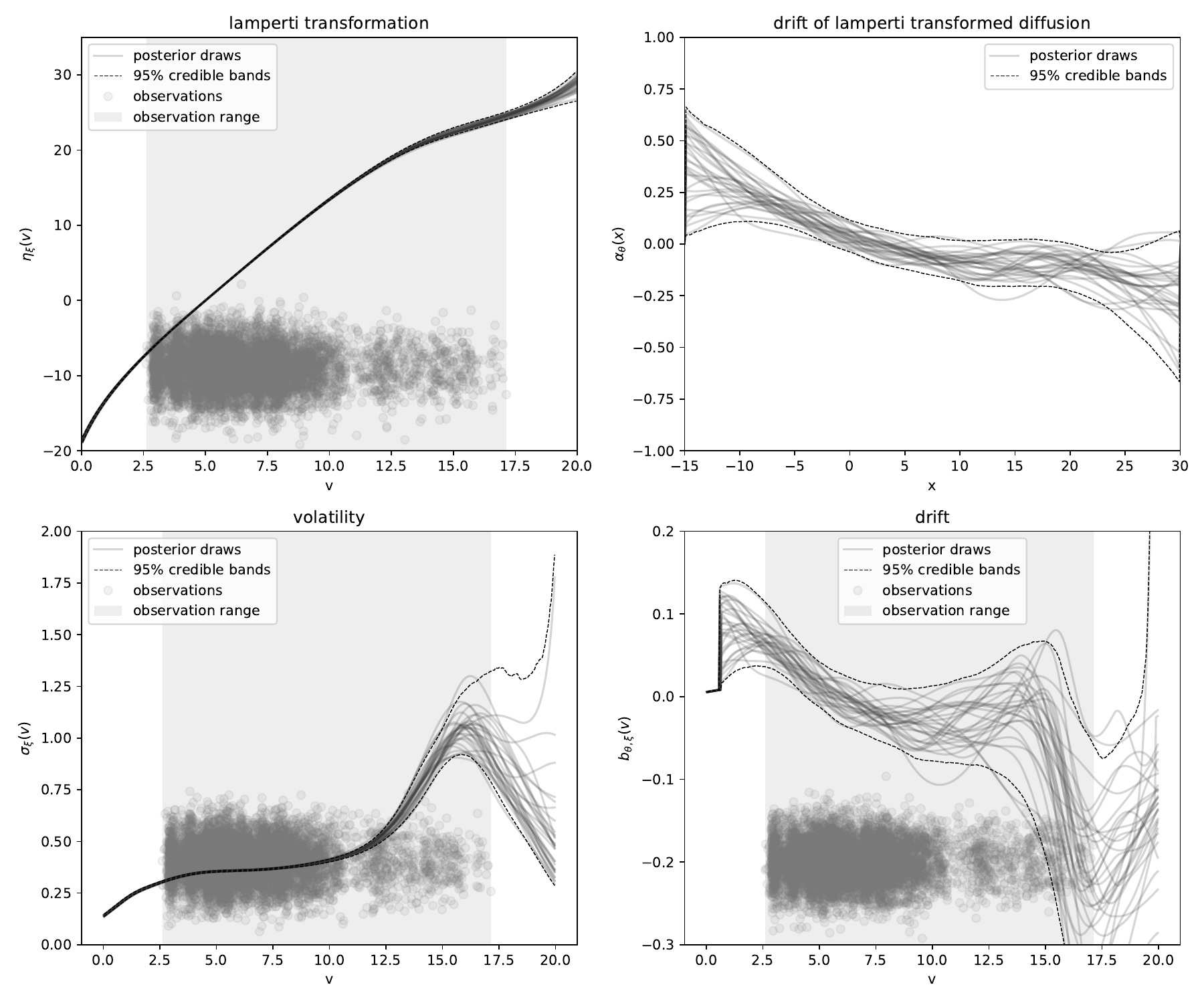}
    \captionof{figure}{Fits to the U.S. Treasury bills dataset. Algorithm \ref{alg:mcmc} was run for 200,000 iterations and 30 random draws from the last 100,000 steps were plotted. The format of the plot is the same as in Figure \ref{fig:toy_example_fits}. \label{fig:finance_fits}}
\end{figure}

We plot the 95\% empirical credible regions for the purpose of visualising the uncertainty regarding presence of non-linearities. We recover the results of \cite{stanton_nonparametric_1997} quite closely, though minor differences are present. Just as in \citet[Figure 5]{stanton_nonparametric_1997} we observe clear evidence for non-linearity of the volatility term, manifesting itself in a rapid increase towards greater values at higher levels of interest rate (Figure \ref{fig:finance_fits}, bottom left panel). This result intuitively makes sense: extraordinarily high interest rates are expected to be observed only during the most uncertain times for the financial markets. Indeed, the very highest interest rates in Figure \ref{fig:finance_example_obs} fall in the late 70s and early 80s---a time of high inflation, contractionary monetary policy, and an ensuing recession of the U.S.\ economy. Unlike \cite{stanton_nonparametric_1997} however, we note that the volatility term appears nearly flat for a range of interest rates: 4--9\%. Additionally, it is apparent that for smaller values of interest rates (2.5--8\%) the drift coefficient acts as a gentle mean reversion term (bottom right panel, compared with \citet[Figure 4]{stanton_nonparametric_1997}). It flattens out at medium interest rates (8--15\%) and then changes to a very strong mean reversion term for large values of interest rates (15\%+), preventing them from exploding to infinity. Relatively stronger mean reversion is required to counteract the increased level of volatility.

\subsection{Paleoclimatology dataset}

In this section we analyse isotopic records from ice cores drilled and studied under the North Greenland Ice Core Project \citep{andersen2004high}. The data consists of the estimates of the historical levels of $\delta^{18}O$ present in Greenland's ice cores during their formation, dating back 123,000 BP (before present) until present. $\delta^{18}O$ is a ratio of isotopes of oxygen (18 and 16) as compared to some reference level (with a known isotopic composition) and is a commonly used measure of the temperature of precipitation.

\begin{figure}[ht]
  \begin{center}
    \includegraphics[width=0.8\textwidth]{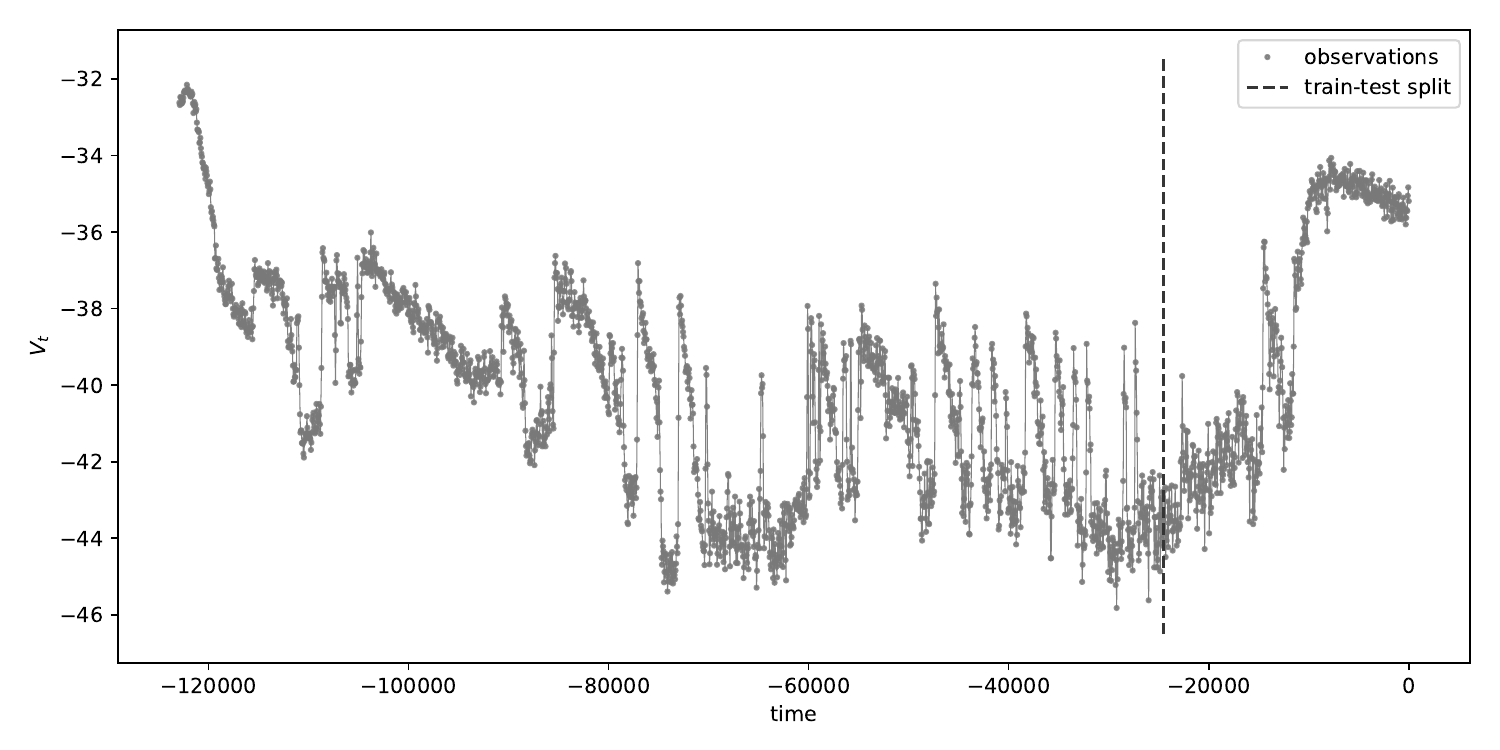}
    \caption{50 year mean values of $\delta^{18}O$ present during formation of ice cores in North Greenland. We remark that these particular data are usually displayed on a reversed time-axis; however, as our diffusion model needs to obey the laws of causality we use the usual convention for time-axis.
 \label{fig:paleoclimatology_example_obs}}
  \end{center}
\end{figure}

The data (Figure \ref{fig:paleoclimatology_example_obs}) shows oscillations between two states (so called Dansgaard-Oeschger (DO) events): stadial (cold) and interstadial (warm), and it exhibits sharp shifts between the two. Presently employed global circulation models are unable to reconstruct this phenomenon, which puts into question some of the conclusions that could be drawn from such models \citep{ditlevsen2009stochastic}. Consequently, one of the scientific goals is to understand the mechanisms causing DO events (see \cite{ditlevsen2009stochastic} and references therein for some hypotheses put forth). SDEs are one of the tools used for this purpose \citep{alley2001stochastic, ditlevsen2007climate}. An example often employed in the literature is a stochastic resonance model, such as a double-well potential, possibly with an additional periodic component in the drift \citep{alley2001stochastic,ditlevsen2005recurrence,ditlevsen2007climate,krumscheid2015data}. The validity of such diffusion models does not yet seem to have reached a consensus. Recently, \cite{garcia_nonparametric_2017} have fitted a non-parametric diffusion model to these data, however the method used by the authors is based on Euler--Maruyama discretization and introduces difficult to quantify bias, which might be substantial. We fit our exact and flexible model with the aim of finding an appropriate family of parametric diffusions, and without a priori assuming the form of a stochastic resonance model.

Again following our methodology, the orders of bases $u$ and $h$ were set to $4$ and $3$ respectively, and knots for $\A$ and $\lamp$ were placed respectively at:
\begin{align*}
  \kappa^{({\theta})}&\leftarrow\{-15_{(5)},-12.5,-10,-7.5,-5,-2.5,0,
   2.5,5,7.5,10,12.5,15,17.5,20,22.5,25_{(5)}\},\\
    \kappa^{({\xi})}&\leftarrow\{-49_{(4)},-47,-45,-43,-41,-39,-37,-35,-33,-31,-29_{(4)}\}.
\end{align*}
$\bar{v}$ was set to $-40$. The regularization parameters were set to $\lambda_{1,3}=5000$, $\lambda_{2,2}=1$ and Algorithm \ref{alg:mcmc} was run for 300,000 iterations. The results are given in Figure \ref{fig:paleoclimatology_fits}.

\begin{figure}[!t]
    \includegraphics[width=\textwidth]{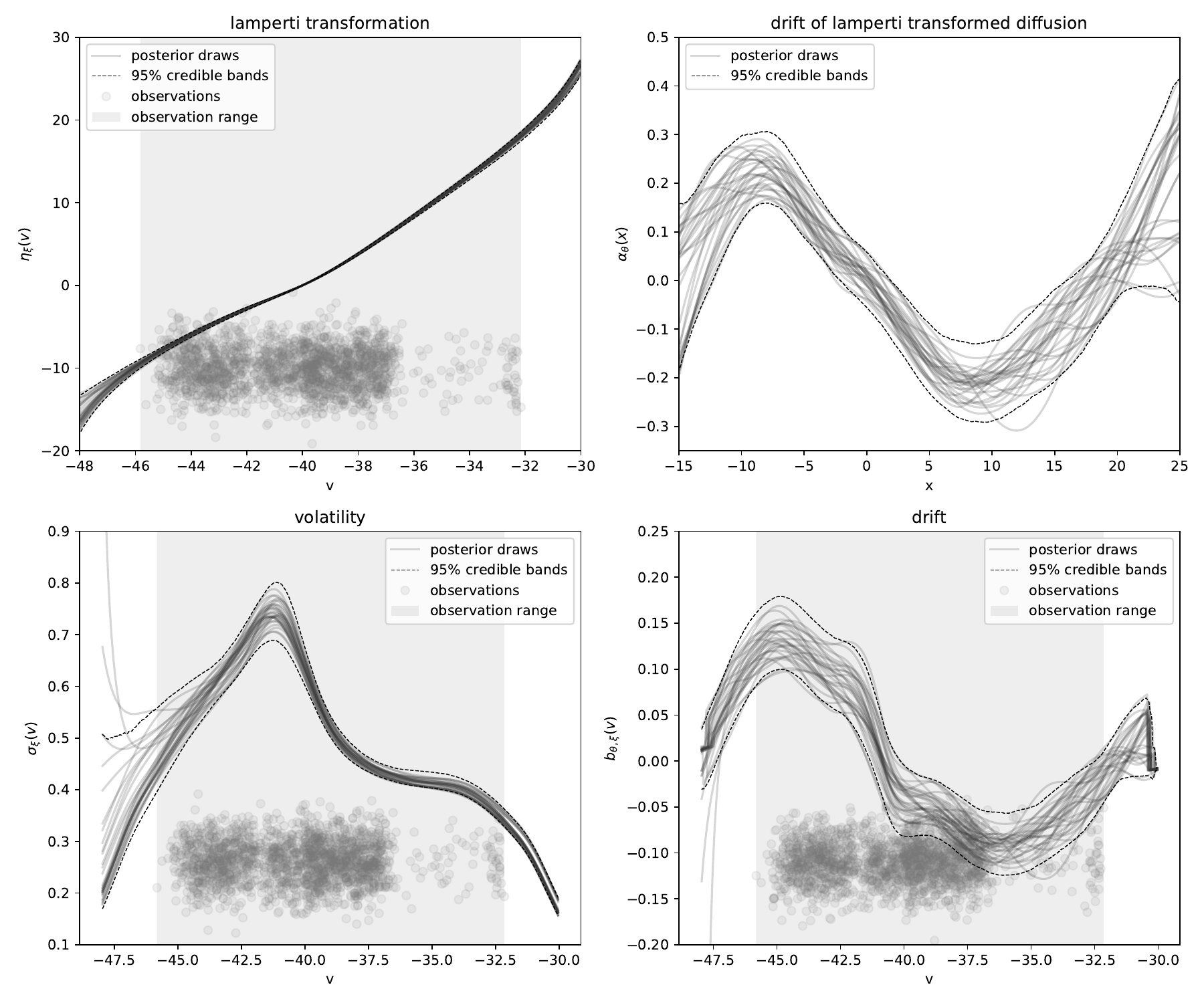}
    \captionof{figure}{Fits to the $\delta^{18}O$ dataset. Algorithm \ref{alg:mcmc} was run for 300,000 iterations and 30 random draws from the last 150,000 steps were plotted. The format of the plot is the same as in Figure \ref{fig:toy_example_fits}. \label{fig:paleoclimatology_fits}}
\end{figure}

The results are somewhat surprising. The drift parameter indeed appears to be consistent with that of a double-well potential model (this behaviour is more pronounced for the drift of a Lamperti-transformed diffusion), producing the observed separation of stadial and interstadial states. However, the volatility coefficient appears to be an equally strong non-linear contributor. It spikes around $-41$, which is the trough between the stadial and interstadial regions, allowing for more frequent transitions between two states than would have otherwise been possible under a regular double-well potential model. Our approach shows that existing stochastic resonance dynamic models are not adequate. This could perhaps be an indication that a richer class of models is needed to explain the NGRIP data, such as via more sophisticated drift and volatility coefficients or via a multi-dimensional diffusion model. 

\subsection{Astrophysics dataset}
Active galactic nuclei (AGNs) are luminous objects sitting at the centres of galaxies. Quasars comprise a subset of the brightest of AGNs. The level of luminosity emitted by those objects varies over time. However, reasons for their variability are unclear \citep{kelly2009variations}.

\begin{figure}[ht]
  \begin{center}
    \includegraphics[width=0.75\textwidth]{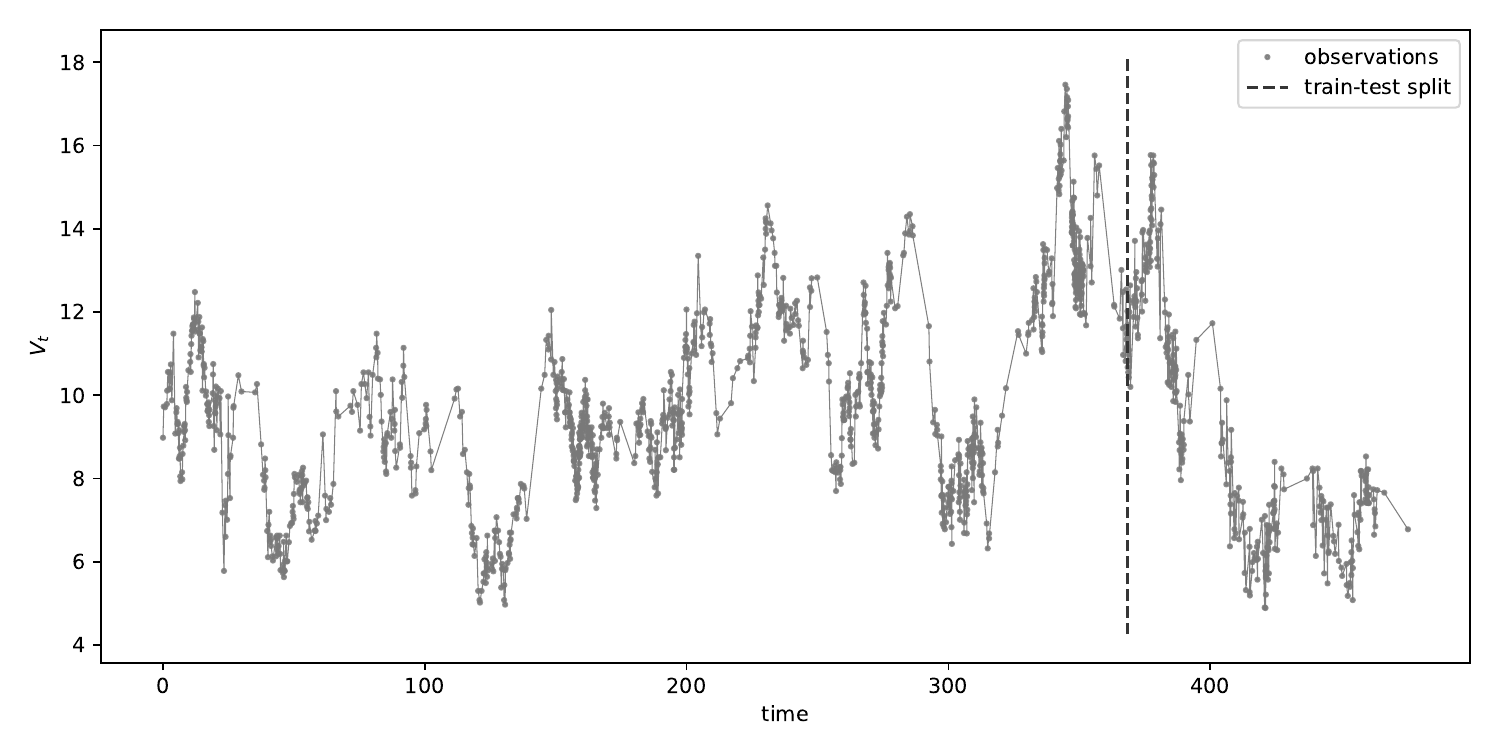}
    \caption{Observations of light curves of NGC 5548 (optical continuum at 5100\AA{} in units of $10^{-15}$ ergs s$^{-1}$ cm$^{-2}$ \AA$^{-1}$). Data taken from AGN Watch Database \citep{agn:ngc5548}. \label{fig:astrophysics_example_obs}}
  \end{center}
\end{figure}

\cite{kelly2009variations} performed a comprehensive study of the optical light curves of quasars, under the assumption that they can be described by an Ornstein--Uhlenbeck process. This choice was dictated not by an understanding of the mechanism governing the phenomenon, but instead by seeking a model exhibiting three properties: (i) a continuous-time process, (ii) consistent with the empirical evidence for spectral density being proportional to $S(f)\propto 1/f^2$, and (iii) parsimonious enough to apply to large datasets.
Additionally, as the authors note, ``much of the mathematical formalism of accretion physics is in the language of differential equations, suggesting that stochastic differential equations are a natural choice for modeling quasar light curves''. Naturally, this raises the question of whether more complex diffusion processes could provide better fits to the data. We fitted our flexible model to an observation of a single quasar NGC 5548 (Figure \ref{fig:astrophysics_example_obs}, from \citet[Fig.\ 4, left]{kelly2009variations}), taken from the AGN Watch Database \citep{agn:ngc5548}. Our aim was to investigate whether fitting a flexible model would exhibit significant deviation from the assumed OU process.

The orders of bases $u$ and $h$ were set to $4$ and $3$ respectively, and knots for $\A$ and $\lamp$ were placed respectively at:
\begin{align*}
  \kappa^{(\theta)}&\leftarrow\{-4.3_{(5)},-3.44, -2.58, -1.72, -0.86, 0, 0.86, 1.72, 2.58, 3.44, 4.3_{(5)}\},\\
    \kappa^{(\xi)}&\leftarrow\{4_{(4)},5.4, 6.8, 8.2, 9.6, 11, 12.4, 13.8, 15.2, 16.6, 18_{(4)}\}.
  \end{align*}
$\bar{v}$ was set to $-40$. The regularization parameters were set to $\lambda_{1,3}=100$, $\lambda_{2,2}=1$,  and Algorithm \ref{alg:mcmc} was run for 400,000 iterations. The results are given in Figure \ref{fig:astro_coefs}. The drift of a Lamperti-transformed diffusion indeed appears to be consistent with a simple mean-reversion term $-a(X-b)$ of the OU process; however, the volatility term exhibits strongly non-linear behaviour. In particular, in two regions of space: 7--9 and 13--15, the volatility is elevated resulting in a high irregularity of the drift of an underlying process. These results strongly suggest that an OU process is too restrictive a model for the given data.

\begin{figure}[!t]
    \includegraphics[width=\textwidth]{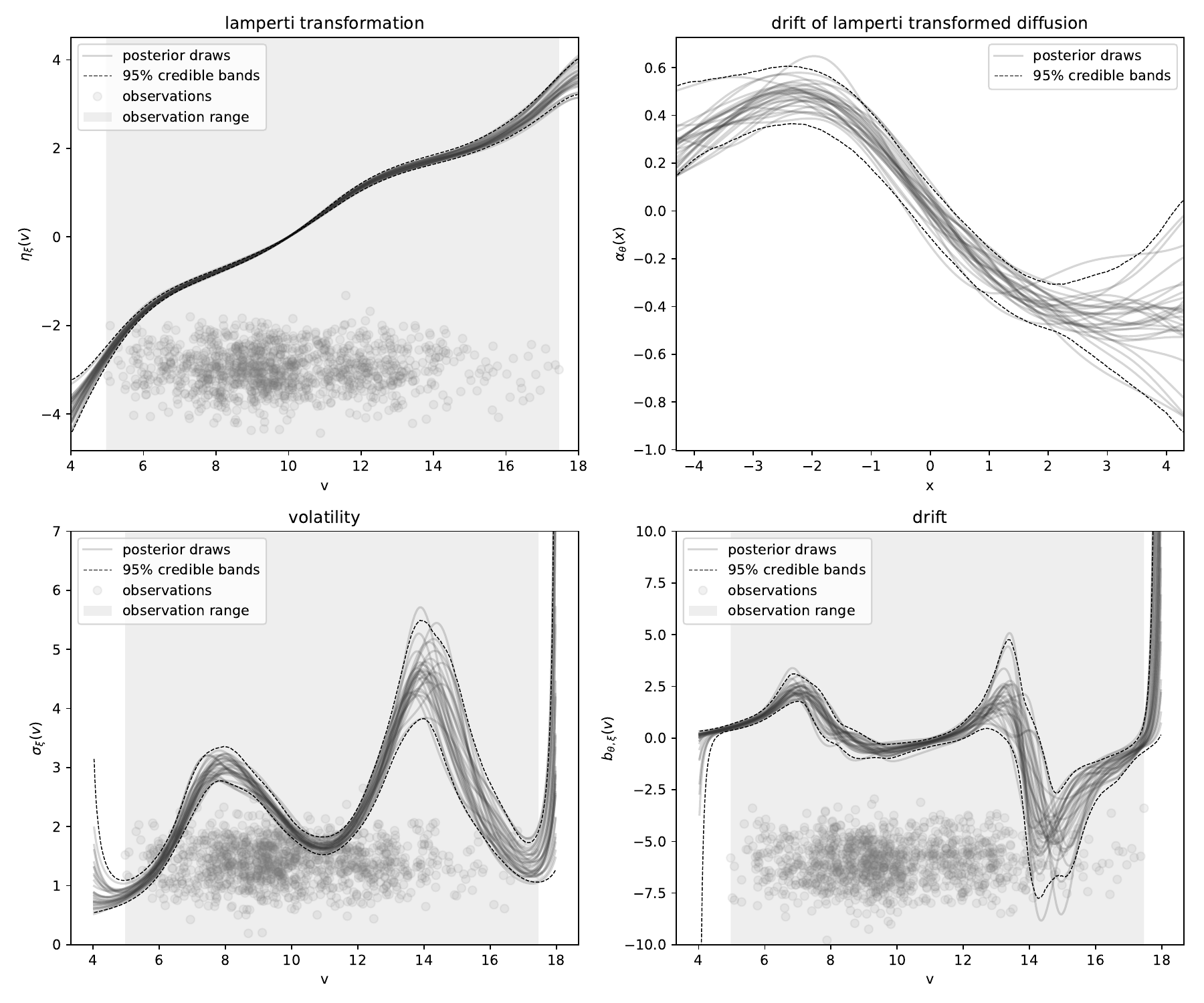}
    \captionof{figure}{Fits to the NGC 5548 dataset. Algorithm \ref{alg:mcmc} was run for 400,000 iterations and 30 random draws from the last 200,000 steps were plotted. The format of the plot is the same as in Figure \ref{fig:toy_example_fits}.\label{fig:astro_coefs}}
\end{figure}

\section{Conclusion} \label{sec:conclusions}
In this article we introduced a flexible Bayesian algorithm for simultaneous inference of both drift and diffusion coefficients of an SDE from discrete observations. The method avoids any time-discretization error and does not make any assumptions about the frequency or spacing of the observations; it is therefore naturally suited to handle time series data with missing observations. Key to our approach is to model indirectly using a spline basis suitable transformations of the drift and diffusion coefficients. We developed, with consideration of practical issues, an MCMC algorithm for sampling from the posterior of the basis parameters, given discrete observations from an SDE. We hope that visualization of potential functional forms for $b$ and $\sigma$ will be a powerful investigative tool for practitioners, allowing them to refine their understanding of the processes they are studying. 

The range of real-world examples considered in the numerical section of this paper demonstrate the breadth of applicability of our methodology. In the illustrative example of Section \ref{sec:illustrative}, we showed that even in the presence of severe over-specification of the number of knots the recovery of true functions is possible. Analysis of the financial dataset resulted in conclusions largely in agreement with what has already been observed in the literature through the use of other, frequentist, non-parametric methods. Our methodology suggests that for the paleoclimatology and astrophysics examples a richer class of statistical models to those currently used by practitioners seems to be needed. This was aided by our ability to visualise plausible posterior functional forms of the drift and diffusion coefficients.

A substantial direction to extend our methodology would be to broaden its applicability from scalar to multi-dimensional processes. Grounds for optimism are that for a $d$-dimensional diffusion of gradient type (i.e.\ for which a potential $\A:\mathbbm{R}^d\rightarrow \mathbbm{R}$ exists), the potential can be modelled through multivariate interpolations of B-splines. However, a potential complication is the modelling of the Lamperti transformation, which must satisfy additional conditions analogous to the monotonicity property required in dimension one \citep[see][]{ait2008closed}. 

Another natural direction would be to develop a companion methodology for noisy observations from an SDE, as opposed to observations without noise \citep{beskos2006exact,beskos2009monte}, a very common framework in many applications.

\section*{Acknowledgments}

The authors would like to thank Dr Marcin Mider for substantial contributions to the development of this work.

\subsection*{Declarations}
\begin{enumerate}
\item All three authors were supported by the Alan Turing Institute under the EPSRC grant EP/N510129/1.  Gareth Roberts was additionally supported under the EPSRC grants EP/K034154/1, EP/K014463/1, EP/R034710/1 and EP/R018561/1.
    \item
The authors: Paul Jenkins, Murray Pollock and Gareth Roberts have no conflicts of interest or competing interests to disclose.
\item
The US Treasury Bill data can be accessed at {\tt https://home.treasury.gov }. Our Paleoclimatology example uses data from \cite{andersen2004high}, and  data for the Astophysics example is taken from AGN Watch Database \citep{agn:ngc5548}.
\item The authors: Paul Jenkins, Murray Pollock and Gareth Roberts, contributed equally to all aspects of this article. 
\end{enumerate}

\appendix
\section{Details of Algorithm 1} \label{apx:supplement}
\subsection{Updating the latent diffusion path}
Here we give further details of Step \ref{algStep:psrs} of Algorithm \ref{alg:mcmc} when we have discrete observations from a diffusion of the form \eqref{eq:master_SDE}. See Section \ref{sec:inference_algo} for some intuition. As noted in that Section, following an application of the Lamperti transformation the datapoints $\{\lamp(v_{t_i}):\: i=1,\dots,N\}$ now depend on a parameter of interest, reintroducing the very problem of degeneracy that the Lamperti transformation was designed to avoid. To resolve this, the general idea is to decouple the dependence of the source of randomness from the parameter $\xi$ when constructing the variable $\mathcal{S}(\mathcal{A})$. To this end, the \emph{innovation process} is defined:
\begin{equation*}
    Z:=\{Z^{[i]}\}_{i=0}^{N-1}\in \bigotimes_{i=0}^{N-1}\mathcal{C}(\mathbbm{R};[0,\Delta_i]), \mbox{ where } \Delta_i:=t_{i+1}-t_i,
\end{equation*}
with $\bigotimes$ denoting a product space and $\mathcal{C}(B;A)$ a space of continuous functions from $A\rightarrow B$ (notice that $Z$ is independent of $\xi$), together with a function:
\begin{equation*}
    \Psi(\cdot;\xi):\bigotimes_{i=0}^{N-1}\mathcal{C}(\mathbbm{R};[0,\Delta_i])\rightarrow\mathcal{C}(\mathbbm{R};[0,T]). 
\end{equation*}
We construct $\Psi$ and $Z$ so that $V^{\circ}:=\Psi(Z;\xi)$ serves as a proposal from $\pi(\mathcal{A}|\theta,\xi,\mathcal{D})$ within a rejection sampler. In practice this can be achieved by instead using a finite-dimensional summary of $V^{\circ}$ as a proposal from $\pi(\mathcal{S}(\mathcal{A})|\theta,\xi,\mathcal{D})$.

In order to construct the pair $(Z,\Psi(\cdot;\xi))$, a diffusion $V$ solving \eqref{eq:master_SDE} is first transformed to a diffusion $X:=\{\lamp(V_t), t\in[0,T]\}$ via the Lamperti transformation \eqref{eq:lamperti}.
The set of transformed observations is then defined as $\widetilde{\mathcal{D}}_{\xi}:=\{\obs_i:\: i=0,\dots,N\}:=\{\lamp(v_{t_i}):\: i=0,\dots,N\}$. Finally, the \emph{centering} functions are defined as:
\begin{equation}\label{eq:centering_fn}
    \m_{i}(t):= \obs_{i} + \frac{t}{\Delta_i}(\obs_{i+1}-\obs_{i}),\quad t\in[0,\Delta_i],\quad i=0,\dots N-1,
\end{equation}
and a map $\varsigma(\cdot;\xi):\bigotimes_{i=0}^{N-1}\mathcal{C}(\mathbbm{R};[0,\Delta_i])\rightarrow\mathcal{C}(\mathbbm{R};[0,T])$ as:
\begin{equation*}
    \varsigma(Z;\xi)_t:=\sum_{i=0}^{N-1}\left(Z_{t-t_i}^{[i]} + \m_{i}(t-t_i)\right)\mathbbm{1}_{(t_i,t_{i+1}]}(t),\quad t\in(0,T],
\end{equation*}
(with $\varsigma(Z;\xi)_0:=V_0$). \citet{roberts2001inference} define the innovation process as draws from the measure $\bigotimes_{i=0}^{N-1}\mathbbm{W}^{(\Delta_i,0,0)}$, and $\Psi(\cdot;\xi)$ to be given by $\Psi(Z;\xi):=\lamp^{-1}\circ \varsigma(Z;\xi)$, where we recall that $\mathbbm{W}^{(t,x,y)}$ denotes the law induced by a Brownian bridge from $x$ to $y$ over the interval $[0,t]$.

To impute the finite-dimensional surrogate variable $\mathcal{S}(\mathcal{A})$ it is now enough to employ $N$ independent path-space rejection samplers, each for a separate interval $(t_i,t_{i+1})$, $i=0,\dots,N-1$, and whenever a proposal path $V^{\circ}$ needs to be revealed at a time-point $t\in(t_i,t_{i+1})$, an innovation process $Z^{[i]}$ is sampled at time $t-t_i$ in order to obtain $V^{\circ}_t=\Psi(Z;\xi)_t$. Since path-space rejection sampling reveals proposals only at a discrete collection of (random) time-points, the simulated innovation process is given by the following proposal surrogate random variable $\mathcal{S}^{\circ}(\mathcal{A})$:
  \begin{equation}\label{eq:surrogate}
    \mathcal{S}^{\circ}(\mathcal{A}):=\{\{Z^{[i]}_{\psi_{i,j}},\{\psi_{i,j}, \chi_{i,j}\}\}_{j=1}^{\varkappa_i}, \varkappa_i,\Upsilon_i\}_{i=0}^{N-1}.
  \end{equation}
Here $\{\psi_{i,j}, \chi_{i,j}\}_{j=1}^{\varkappa_i}$ is a Poisson Point Process on $[0,\Delta_i]\times[0,1]$ with intensity $\upBd(\Upsilon_i)$ (where $\upBd(\Upsilon_i)$ is a local upper bound on $\G-\lowBd$, with $\G$ and $\lowBd$ as defined in Section 2). $\Upsilon_i$ is an additional random element containing information about the path $Z$ enabling the upper bound $\upBd$ to be computed (for instance an interval which constrains a Brownian bridge path, $Z^{[i]}$, to a given interval; see \cite{beskos2008factorisation} for more details). This proposal is then accepted with probability proportional to the Radon--Nikodym derivative between the proposal and the target laws, as described in Section \ref{sec:inference_algo}.

\subsection{Updating the parameters}
With the surrogate defined as in \eqref{eq:surrogate}, \citet{sermaidis_markov:2013} derive a closed form expression for the joint density of the imputed data $\mathcal{S}(\mathcal{A})$, parameters $(\theta,\xi)$ and observations $\mathcal{D}$, which becomes:
\begin{equation}
  \label{eq:joint_param_sampler}
  \begin{split}
    \jointDist{}{}&= \pi(\theta,\xi)\exp\bigg\{ \A(\obs_T) - \A(\obs_0) - (\lowBd - 1)T - \sum_{i = 0}^{N-1}\upBd(\Upsilon_i,\m) \bigg\}\\
    &\phantom{=}\quad\cdot\prod_{i=0}^{N-1}\bigg\{\mathcal{N}_{\Delta_i}(\obs_{i+1}-\obs_{i})D_{\xi}(v_{t_{i+1}})\upBd(\Upsilon_i,\m_i)^{\varkappa_i}\\
    &\phantom{=}\qquad\qquad\cdot\prod_{j=1}^{\varkappa_i}\bigg[1-\big(\G(Z^{[i]}_{\psi_{i,j}}+\m_{i}({\psi_{i,j}}))-\lowBd\big)/\upBd(\Upsilon_i,\m_i)\bigg]\bigg\}.
  \end{split}
\end{equation}
Here we denote by $\mathcal{N}_t(x)$ as a Gaussian density with variance $t$ and mean $0$ evaluated at $x$, and $D_{\xi}(\cdot):=(\vola(\cdot))^{-1}=\lamp'(\cdot)$. Additionally, recall from \eqref{eq:lamperti} the notation for the anti-derivative of the drift of $X$.

The density \eqref{eq:joint_param_sampler} can then be used to sample from $\pi(\xi,\theta|\mathcal{S}(\mathcal{A}),\mathcal{D})$ in the parameter-update step. In particular, to sample from $\pi(\theta,\xi|\mathcal{S}(\mathcal{A})^{(n-1)},\mathcal{D})$ at the $n$th iteration of the Markov chain we can draw $(\theta^{\circ},\xi^{\circ})\sim q((\theta^{(n-1)},\xi^{(n-1)}),\cdot)$ from some proposal kernel (we did this as per \eqref{eq:qstep1} and \eqref{eq:qstep2}, but could be another proposal such as a random walk) and then employ a Metropolis--Hastings correction with the acceptance probability given by Algorithm \ref{alg:mcmc} Step \ref{algStep:mhprob}.

\bibliographystyle{agsm}

\bibliography{Bibliography-MM-MC}
\end{document}